\documentclass[11pt,letterpaper]{article}
\usepackage{amsthm,color,latexsym,graphicx}
\usepackage[hyphens]{url}
\usepackage[margin=1in]{geometry}
\usepackage[font=small,labelfont=bf]{caption}
\usepackage[labelformat=simple]{subcaption}

\definecolor{nicegreen}{RGB}{0,194,0}
\definecolor{nicegrey}{RGB}{153,153,153}

\DeclareFontFamily{U}{matha}{\hyphenchar\font45}
\DeclareFontShape{U}{matha}{m}{n}{
<5> <6> <7> <8> <9> <10> gen * matha
<10.95> matha10 <12> <14.4> <17.28> <20.74> <24.88> matha12
}{}
\DeclareSymbolFont{matha}{U}{matha}{m}{n}

\DeclareMathSymbol{\Lt}{3}{matha}{"CE}
\DeclareMathSymbol{\Gt}{3}{matha}{"CF}

\usepackage{float}

\usepackage{indentfirst,csquotes}

\usepackage{paralist,etoolbox}

\usepackage{
		graphicx,		
		amsmath,		
		amssymb,		
		amsthm,			
		tikz,			
		multicol,		
		color,			
		thmtools,		
		etoolbox,		
		xspace,			
        hyperref,		
        xcolor,         
        datetime,       
        csquotes,        
        svg,
        placeins
        }
        
\usepackage{hyperref}
\hypersetup{breaklinks,
bookmarksnumbered,bookmarksopen,bookmarksopenlevel=2}
{\makeatletter \hypersetup{pdftitle={\@title}}}

\newtheorem{theorem}{Theorem}[]
\newtheorem*{theorem-no-number}{Theorem}
\newtheorem{definition}[theorem]{Definition}
\newtheorem{example}[theorem]{Example}
\newtheorem{lemma}[theorem]{Lemma}

\newtheorem{corollary}[theorem]{Corollary}

\newtheorem{remark}[theorem]{Remark}
\newtheorem{claim}[theorem]{Claim}

\definecolor{nicegreen}{RGB}{0,194,0}
\definecolor{nicegrey}{RGB}{153,153,153}

\title{Hive is PSPACE-Hard} 

\author{%
  Dani\"el I. Andel
\and
  Benjamin G. Rin%
    \thanks{Corresponding author. Utrecht University, Utrecht, The Netherlands.
      \protect\url{b.g.rin@uu.nl}}
    }

\date{June 01, 2025}

\begin{document}

\maketitle

\begin{abstract}
Hive is an abstract strategy game played on a table with hexagonal pieces. First published in 2001, it was and continues to be highly popular among both casual and competitive players. In this paper, we show that for a suitably generalized version of the game, the computational problem of determining whether a given player in an arbitrary position has a winning strategy is PSPACE-hard. We do this by reduction from {\sc Formula Game,} which is first reduced to an intermediate problem we call {\sc Formula Game Geography}, after which the latter is reduced to our decision problem.
\end{abstract}

\bigskip

\section*{Acknowledgments}
We would like to thank Rosalie Iemhoff and Robert Barish for their helpful conversations and comments. We also wish to thank the anonymous referees for their valuable feedback.

\bigskip

\section{Introduction} \label{sec:introduction}

Hive is a popular strategy game designed by John Yianni and published in 2001. Like classic games such as chess, checkers, and Go, it is a two-player, perfect-information, deterministic, turn-based tabletop game with at most one winner. 
In this paper, we
consider the computational complexity of the decision problem {\sc Hive}, which asks whether a given player has a forced winning strategy from a given position in a generalized version of Hive.
Our main result is that this decision problem is PSPACE-hard. We prove this by reduction from {\sc Formula Game}, first by reducing the latter to a variant of {\sc Generalized Geography} that we call {\sc Formula Game Geography (FGG)}, then in turn reducing FGG to {\sc Hive}.\footnote{The proof presented in this paper originates from 
work done for the bachelor thesis of the alphabetically first 
author, written under supervision of the second~\protect\cite{Andel}.}

While the standard version of Hive has only 22 pieces and a small set of rules, it is nevertheless a deeply complex and
strategic game. Still, it has only a finite number of possible game positions, making it possible in principle to construct a large look-up table which unveils a winning strategy. 
Therefore, as with most board games, investigating its asymptotic computational complexity requires us to generalize the game in some natural way. Typically this is done by generalizing 
the board to dimensions~$n \times n$ (cf.~\cite{FraenkelLichtenstein,Hearn2009,go,gobang,Rin-Schipper24,robson,storer}). However, Hive does not have a board; the playing field is determined by the placement of pieces. Hence, we generalize Hive to \emph{$n$-Hive} 
by stipulating simply that each player has an equal set of~$n$ pieces to play with, rather than 11. We make no further adjustments to the rules (Section~\ref{subsec: rules}).

We note, though, that our generalization of Hive to~$n$-Hive allows for each player to have multiple Queens. In Hive, the Queen is a piece that serves a role similar to that of the King in chess, so this generalization is arguably not as natural as one permitting only one Queen. By contrast, complexity bounds for $n\times n$ chess found in~\protect\cite{FraenkelLichtenstein} and~\protect\cite{storer} work with generalizations of chess with only one King per player. We return to this consideration in Section~\protect\ref{sec: conclusion}. 

Hereafter, we may refer to~$n$-Hive  as ``Hive'' for simplicity. Our main theorem is then stated as follows.
\begin{theorem-no-number}[Main theorem]
The problem {\sc Hive} of determining whether a given player has a winning strategy in an arbitrary Hive position is PSPACE-hard.
\end{theorem-no-number}

Analogous results have been found over the last several decades for various other strategy
games. In the 1980s, generalized versions of classic games such as Go~\cite{go}, Gobang~\cite{gobang}, Hex~\cite{hex}, chess~\cite{FraenkelLichtenstein,storer}, and checkers~\cite{robson} were proved to be PSPACE-hard. Since then, the complexity of numerous other games has been determined. Recent examples of PSPACE-hard problems include generalizations of the video game Lemmings~\cite{viglietta}, the ancient Hawaiian board game Konane~\cite{Hearn2009}, and the 2003 board game Arimaa~\cite{Rin-Schipper24}. Some of the known hard games have been found to be PSPACE-complete, others are known to be EXPTIME-complete, and others still have no known upper complexity bound. The determining factor in whether a two-player game's complexity lies within PSPACE is typically the presence or absence of a polynomial bound on the potential length of games (see, e.g.,~\cite[Ch. 6]{Hearn-Demaine}). In the case of Hive, no such bound is known to exist. Accordingly, we consider it likely that {\sc Hive} is not in PSPACE. 

To prove that {\sc Hive} is PSPACE-hard, we begin with the classic {\sc Formula Game} (Section~\ref{subsubsec: Formula Game}), which reduces to the problem {\sc Formula Game Geography (FGG)} we define in Definition~\ref{def: FGG}. Briefly, FGG is {\sc Generalized Geography} played on a directed graph that simulates a {\sc Formula Game} instance, 
with some added constraints on the degrees of nodes.
It is 
straightforwardly seen to be PSPACE-complete~(Corollary~\ref{cor: FGG PSC}). Our  
proof for~{\sc Hive} then follows the same broad approach as that of Lichtenstein and Sipser's well-known proof that $(n\times n)$ Go is PSPACE-hard~\cite{go}.
However, we find that implementing suitable gadgets in Hive is  
nontrivial, both locally in terms of 
each gadget's internal construction and globally 
with respect to the interconnections between 
gadgets.

First, we must account for Hive's unusual geometry. Unlike Go, chess, and most other classic games, Hive does not have pieces interacting on a square grid at convenient~$90^\circ$ angles.  
Instead, Hive pieces are hexagonal and  
always 
placed in 
such a way 
that 
every piece 
borders at least one other. While this 
is not prohibitively limiting, it does call for careful measures in gadget design, as well as
the construction of special gadgets 
with names such as~\emph{60$^\circ$ direction changer} and~\emph{120$^\circ$ direction changer}.

Second, with respect to more global positional features, we encounter challenges arising from our use of 
an important game rule called the
One Hive rule. This 
rule states that the set of all pieces in the position -- collectively called the \emph{Hive} -- must always be connected. Our reduction takes repeated advantage 
of this rule
to make many 
pieces immobile, 
arranging them 
so that 
they would illegally separate the Hive if moved. While this idea is locally 
fruitful in enabling us to limit the players' options,
it also 
demands special care in 
our handling of
gadget connections, 
in order 
to ensure 
that 
the position remains legal under the One Hive rule 
while also keeping the desired pieces
immobile. (See Section~\ref{subsec: remarks one Hive rule}.)

\enlargethispage{1.9\baselineskip} 
We note that our proof does not provide any nontrivial upper bound on {\sc Hive}'s
complexity. As remarked above, we consider it unlikely that {\sc Hive} is in PSPACE. Since the rules do not contain an analogue of the 50-move rule in chess to limit potential game length, we conjecture that {\sc Hive} is in fact EXPTIME-complete. However, we leave the resolution of this conjecture as a matter for future research 
(see Section~\ref{sec: conclusion}).

\subsection{Game rules} \label{subsec: rules}
The official rules of Hive can be found in~\cite{hive-rules}. The game is played on a flat surface by two players, \emph{Black} and \emph{White}. Both players have~eleven hex-shaped tile pieces at their disposal: three Ants, three Grasshoppers, two Beetles, two Spiders, and a Queen Bee. Each piece moves differently. The goal of the game is to surround the opponent's Queen Bee (hereafter usually ``Queen''). Players take turns either placing a new piece on the table or moving one that is already in play. If a player is not able to place or move a piece, the player passes. The game ends when a Queen Bee is surrounded by other pieces, upon which the player whose Queen is surrounded loses the game.\footnote{This wording implies that in a game of $n$-Hive with multiple Queens per player, a player loses when at least one of their Queens is surrounded. One could alternatively define a variant in which a player loses only when \emph{all} their Queens are surrounded. Or, similarly, one could stipulate that a Queen is removed from the game once it is surrounded, and the goal of the game is then to remove all opposing Queens. (This suggestion is due to an anonymous referee.) While we do not treat such variants in this paper, we expect that they, too, are computationally hard, though they require a different reduction.\label{fn: all Queens}}  
If both players would lose simultaneously, the game is a draw. In the proof we only consider positions in which all pieces are on the table. Hence we omit the rules for placing pieces. We first discuss the pieces and their movement rules. Then we will explain the One Hive rule and Freedom To Move rule.

\subsubsection{Queen Bee} \label{subsubsec: queen}
The Queen Bee, or Queen, moves one space per turn. See Figure~\ref{fig: queen} for an example.

\begin{figure}[htbp]
\centering
\includegraphics[scale=0.38,draft=false]{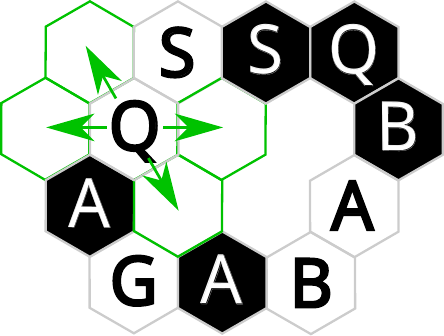}
\caption{The white Queen can move to any one of the four indicated spaces.}
\label{fig: queen}
\end{figure}

\subsubsection{Beetle} \label{subsubsec: beetle}
Like the Queen Bee, the Beetle moves one space per turn. Additionally, the Beetle can, unlike any other piece, move on top of another piece (either color). A piece with a Beetle on top of it is unable to move. Once on top of the Hive, the Beetle can move from tile to tile or drop off again. See Figure~\ref{fig: beetle}.

\begin{figure}[htbp]
\centering
\includegraphics[scale=0.38,draft=false]{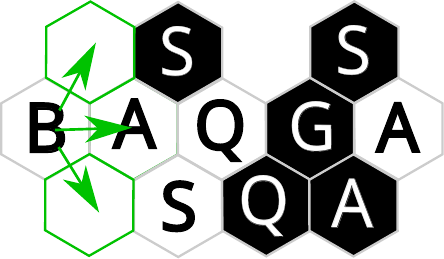}
\caption{The Beetle can move to one of the two open spaces or climb on top of the Ant.}
\label{fig: beetle}
\end{figure}

\subsubsection{Ant} \label{subsubsec: ant}
The Ant can move any number of spaces to any place it can reach, provided it follows the One Hive rule and Freedom To Move rule (Sections~\ref{subsubsec: one hive} and~\ref{subsubsec: freedom}). See Figure~\ref{fig: ant}.

\begin{figure}[htbp]
\centering
\includegraphics[scale=0.38,draft=false]{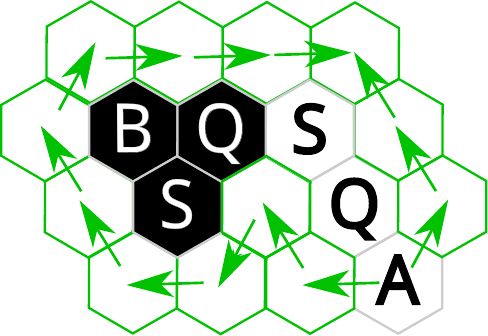}
\caption{The Ant can move to any one of the indicated spaces this turn.}
\label{fig: ant}
\end{figure}

\subsubsection{Grasshopper} \label{subsubsec: grasshopper}
The Grasshopper jumps over pieces. It starts its move by jumping on top of an adjacent piece, then keeps moving in the same direction until it comes off the Hive again. 
See Figure~\ref{fig: grasshopper}.

\begin{figure}[htbp]
\centering
\includegraphics[scale=0.38,draft=false]{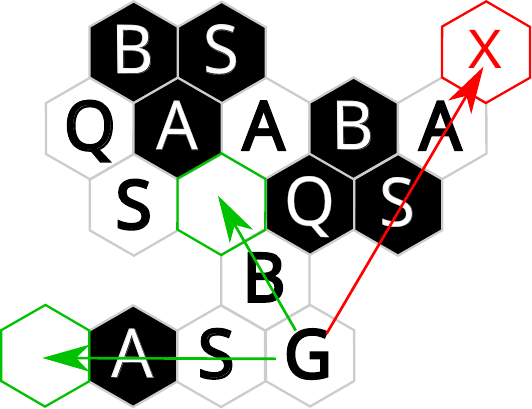}
\caption{The Grasshopper can jump to one of the two indicated spaces. The space marked \textcolor{red}{X} is unavailable, as there is a blank spot between.}
\label{fig: grasshopper}
\end{figure}

\subsubsection{Spider} \label{subsubsec: spider}
The Spider moves exactly three spaces per turn. During this movement it cannot go back and forth between two spaces. An example is in Figure~\ref{fig: spider}.

\begin{figure}[htbp]
\centering
\includegraphics[scale=0.45,draft=false]{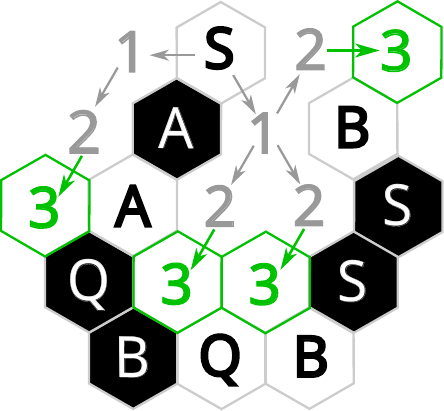}
\caption{The white Spider can move to one of the four spaces marked~\textcolor{nicegreen}{3}. By the One Hive rule (see Section~\protect\ref{subsubsec: one hive}), it may not start its movement by going directly to the space on its right marked~\textcolor{nicegrey}{2} and continue from there, as the Spider would lose touch with the Hive while in transit.}
\label{fig: spider}
\end{figure}

\subsubsection{One Hive rule} \label{subsubsec: one hive}
The pieces in play must at all times be connected. They can be seen to form one (big) Hive. Even \emph{during} a turn -- that is to say, while a piece is in transit sliding toward its destination -- the Hive may not be disconnected and no piece may be left stranded. See Figures~\ref{fig: spider} and~\ref{fig: onehiverule} for examples of movements forbidden under the One Hive rule. 

\begin{figure}
\centering
\subcaptionbox{The white Ant is not able to move, since this would split the Hive in two.\label{fig: one Hive 2}}
[.4\textwidth]{\includegraphics[scale=0.45,draft=false]{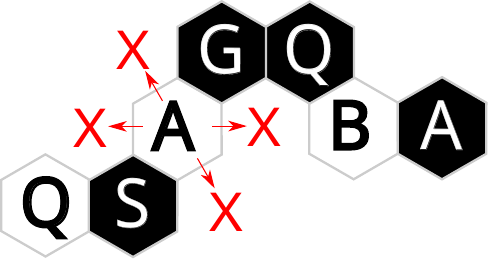}}
\hfill
\subcaptionbox{The Beetle may not move to the space marked \textcolor{red}{X}, as this would leave the Spider stranded while the Beetle is in transit.\label{fig: one Hive}}
[.4\textwidth]{\includegraphics[scale=0.45,draft=false]{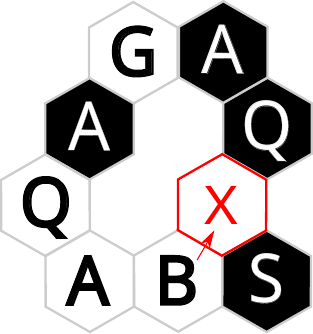}}
\caption{Examples of movement restrictions under the One Hive rule.}
\label{fig: onehiverule}
\end{figure}

\subsubsection{Freedom To Move} \label{subsubsec: freedom}
\enlargethispage{1.6\baselineskip} 
The pieces move in a sliding fashion. So, if a piece is (almost) surrounded and cannot slide physically out of its spot without displacing other pieces, it is not allowed to move. Consider the white Queen in Figure~\ref{fig: beetle}. Recall that this piece is hex-shaped; it is not just the letter ``Q'' written on a hex-shaped location. Since the entire hex is one piece, it is physically stuck and unable to slide out. See Figure~\ref{fig: freedom to move examples} for further examples. Note, however, that Grasshoppers and Beetles provide exceptions to the Freedom To Move rule, as they are able to jump over and climb on top of the Hive, respectively. 

\begin{figure}[htbp]
\centering
\includegraphics[scale=0.45,draft=false]{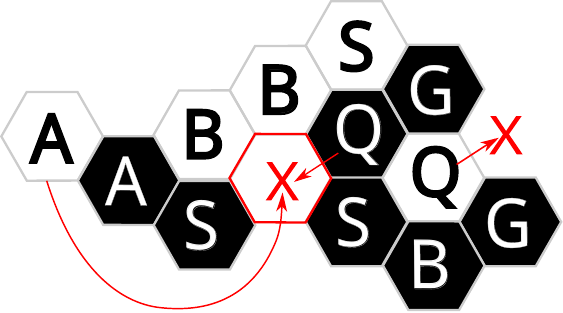}
\caption{The black Queen is unable to move, as it cannot physically slide \emph{out of} its space. The same holds for the white Queen. The white Ant is unable to move to the space marked \textcolor{red}{X}, as it cannot physically slide \emph{into} that space.}
\label{fig: freedom to move examples}
\end{figure}

\subsection{Proof tools} \label{subsec: proof tools}
Here
we discuss some strategic concepts and other tools used within the proof. 

\begin{enumerate}

\item \textbf{Win Threats} \ The most obvious strategic concept is that of a win threat. A piece threatening a win-in-one must be stopped.  See Figure~\ref{fig: trapping} for an example of a win threat being countered by locking down the dangerous piece using the One Hive rule (our next proof tool).

\item \textbf{Locking down pieces using the One Hive rule} \ In our proof it is useful to have a player immobilize an opposing piece by moving another piece away. Figure~\ref{fig: trapping} shows an example in which Black is to move and must immobilize a dangerous white Grasshopper.

\begin{figure}[htbp!]
\centering
\includegraphics[scale=0.35,draft=false]{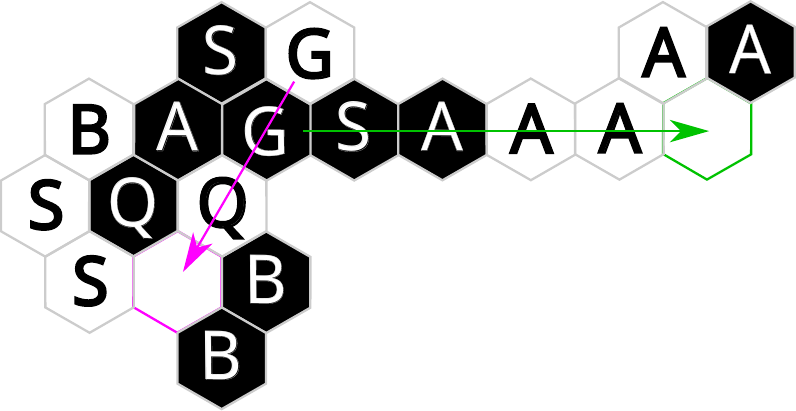}
\caption{White's threat here is the win-in-one by the Grasshopper; moving it to the indicated space would complete the surrounding of the black Queen. Black can immobilize this white Grasshopper by moving the black Grasshopper away. Then the white Grasshopper is the only link between the left and right sides of the Hive, making it no longer allowed to move. As an aside, recall from Section~\ref{subsubsec: grasshopper} that even if it were still allowed to move,  
it would need the black Grasshopper in its original 
space anyway in order to jump over and reach the winning space.}
\label{fig: trapping}
\end{figure}

In general, we call a piece \emph{locked down,} or simply \emph{locked,} when it is either unable to move (by the game rules) or not rational to move (because the piece's owner would lose the game soon after moving it). 

\item \textbf{Freeing a piece} In our proof, it will be useful to construct positions containing large numbers of locked pieces, to keep the position simple and control the flow of the game. A piece that is locked down may sometimes become unlocked (or \emph{free}) if the circumstances nearby change. When a locked-down piece has no prospect of ever becoming free in the future, 
we call it a \emph{dead} piece (see~below).

All pieces in the proof begin locked down except for one. The position will be arranged so that moving one piece will free another, which in turn frees the next piece while locking the previous one, and so on. Such arrangements can result in a forced sequence, or at least a sequence with a very limited number of choices. The freeing of a piece is usually done by giving the Hive a secondary path of connection so that the piece in question no longer needs to stay still just to keep the position legal under the One Hive rule. Examples occur in every gadget throughout the proof. In other cases, a piece is trapped by the Freedom to Move rule, until one of the pieces surrounding it is moved away.

\item \textbf{Selfmate hex} \ A hex is a \emph{selfmate} hex if a piece entering would complete the surrounding of a Queen of its own color. With careful placement of selfmate hexes in our construction, we can restrict the movement of some pieces to help control the flow of play. Selfmate hexes are especially useful
to restrict pieces whose available spaces are few and far apart, such as the Spider and Grasshopper. See Figure~\ref{fig: selfmate tool} for an example with a Grasshopper.
\begin{figure}[htbp!]
\centering
\includegraphics[scale=0.35,draft=false]{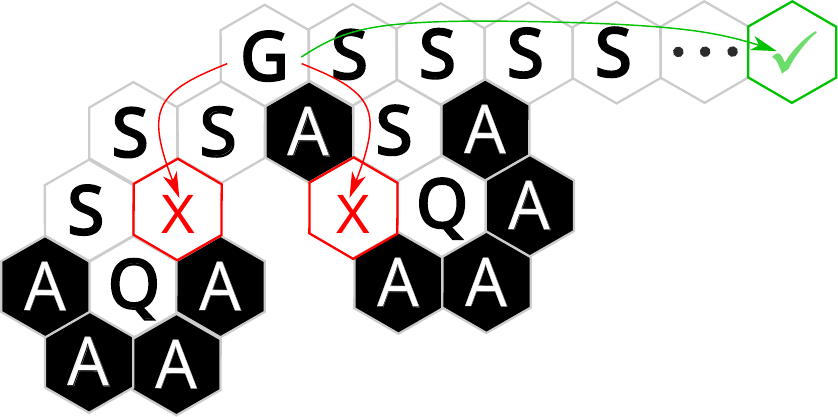}
\caption
{A position with two selfmate hexes for White, marked \textcolor{red}{X}. The Grasshopper cannot safely jump to these, as doing so surrounds a white Queen, losing immediately.}
\label{fig: selfmate tool}
\end{figure}

\item \textbf{Beetle tower} \ The Beetle is the one piece that can climb on top of the Hive, even on top of another Beetle that is already on top of the Hive. In fact, there is no limit to how high the Beetle can climb. It is thus possible to build a Beetle tower as high as desired, with any type of piece on the bottom.

\enlargethispage{\baselineskip}
\item \label{item: dead pieces} \textbf{Dead piece} \ As mentioned, a dead piece is one that is either illegal or fatal to move. There are multiple ways to construct a dead piece, all with benefits and drawbacks. The list below is not exhaustive, but these are the dead pieces used in the proof.
\begin{enumerate}
\item \emph{Dead interior piece:} a Spider or Ant that is completely or almost completely 
surrounded, placed so it will remain
this way until the end of the game. The color of this dead piece is irrelevant. Such pieces are labeled \textbf{DI} in diagrams.
\item \emph{Dead exterior piece:} a white Beetle on top of six black Beetles, with a white Queen on the bottom. (There is also a version with colors reversed.) This dead piece 
is 
usable in places where the top Beetle cannot meaningfully affect the game within one space's reach. If the white Beetle moves, the black Beetle underneath can immediately climb on top of it and lock it down. As moving the white Beetle made no immediate threat, Black can now surround the helpless white Queen underneath in at most five moves. Such pieces are labeled \textbf{DE} in diagrams, colored with the top Beetle's color. 

To implement this dead piece, we must take special care in our proof construction to ensure that the white Beetle's movement indeed achieves nothing strategically relevant, such as freeing a dangerous piece that can
win the game faster than the black Beetles can. 
If there would be any place where this condition fails, in principle it is possible to implement other dead piece designs -- for example, a white Beetle covering a black Ant that, if uncovered, could immediately reach a white Queen to surround. 
\item \emph{Special dead piece:} a Beetle on top of three enemy Beetles and an enemy Queen below. Inherently, it is not necessarily dead, but in the context where it is used, the top Beetle is strategically unable to move. Such pieces, labeled \textbf{DS}, are discussed further in Section~\ref{subsec: quantifier}.  \label{special dead piece}

\end{enumerate}

\item \textbf{Chains of locked-down pieces} \ Any number of pieces can be immobilized under the One Hive rule by placing them in a long chain with a dead exterior piece at the end. 

\end{enumerate}

\section{Proof preliminaries}\label{sec: preliminaries}

The material in this section is mostly standard (see, e.g.,~\cite[Ch.~8] {sipser2013}).

\subsection{Quantified Boolean Formulas} 

\begin{definition}
The set \emph{QBF (fully quantified Boolean formulas)} is the set of all syntactic strings $Q_1 v_1 \dots Q_n v_n \psi$, where $Q_i \in \{ \forall,\exists \}$, each $v_i$ is a Boolean variable, and $\psi$ is a Boolean formula in conjunctive normal form containing only variables from $\{v_1,\dots,v_n\}$.
\end{definition}
\begin{example} The formula $\exists v_1 \forall v_2 \exists v_3 \forall v_4 ((v_1 \lor \neg v_2 \lor v_3) \land (v_2 \lor \neg v_4))$ is in QBF.
\end{example}

\begin{definition}
The truth-value of any $Q_1 v_1 \dots Q_n v_n \psi \in$~\emph{QBF} is defined recursively by:

$\exists v_1\ Q_2 v_2 \dots Q_n v_n \psi$ is true if $Q_2 v_2 \dots Q_n v_n \psi$ is true for some possible truth-value for $v_1$.

$\forall v_1\ Q_2 v_2 \dots Q_n v_n \psi$ is true if $Q_2 v_2 \dots Q_n v_n \psi$ is true for each possible truth-value for $v_1$.
\end{definition}

\begin{definition}

{\sc TQBF} $= \{\langle \varphi \rangle \mid \varphi\  \text{is a true fully quantified Boolean formula}\}$. 
\end{definition}

\begin{theorem}
{\sc TQBF} is PSPACE-complete. 
\end{theorem} 

\subsubsection{Formula Game} \label{subsubsec: Formula Game} 
Formula Game is a two-player game played between a so-called $\exists$-player and $\forall$-player with a given formula~$Q_1 v_1 \dots Q_n v_n \psi \in$~QBF. Players take turns choosing truth-values for the variables~$v_1,\dots,v_n$ in order, with the type of the $i^\text{th}$ quantifier~$Q_i$ determining which player chooses the value of~$v_i$. After~$v_n$ is valuated, the $\exists$-player wins if the formula is true; otherwise, the $\forall$-player wins. Alternatively, since~$\psi$ is in conjunctive normal form, we could equivalently continue the game for two more turns: one in which the $\forall$-player chooses a clause~$c$, and one in which the $\exists$-player chooses a literal~$l$ in~$c$. The $\exists$-player then wins if and only if~$l$ is true. In what follows, we assume this alternate version of the game.

\begin{theorem}
The problem {\sc Formula Game} of determining whether the $\exists$-player has a winning Formula Game strategy for a given formula is PSPACE-complete.
\end{theorem}

\begin{remark}
It is straightforward to see that every formula in QBF is equivalent to one with the form $\exists v_1 \forall v_2 \exists v_3 \forall v_4 \dots \exists v_n \psi$ -- that is, a formula whose quantifiers alternate and whose first and last are both $\exists$. We can easily construct it by adding quantifiers where needed that bind 
extra variables not occurring in $\psi$. 
We therefore have the following result.
\end{remark}

\begin{theorem}\label{thm: FG PSC}
{\sc Formula Game} and {\sc TQBF} remain PSPACE-complete even when restricted to formulas of the form $\exists v_1 \forall v_2 \exists v_3 \forall v_4 \dots \exists v_n \psi$.
\end{theorem}

\subsubsection{Generalized Geography} 

\begin{definition}\label{def: GG}
\emph{Generalized Geography} is a two-player game played on the nodes of a directed graph.
A designated \emph{start} node begins the game marked. Players then alternate turns by marking an unmarked node
that has an incoming edge from the last marked node. A player to move with no unmarked nodes available (all available choices are already marked) loses the game. {\sc Generalized Geography (GG)} is the corresponding decision problem.
\end{definition}
The proof of the theorem below is very well known (see, e.g.,~\cite[Sec. 8.3]{sipser2013}), but we present it here anyway as it will aid with the exposition of our main proof.
\begin{theorem}\label{thm: GG PSH}
{\sc GG} is PSPACE-complete.
\end{theorem}
\begin{proof}
We omit the proof that GG is in PSPACE. We give a polynomial-time reduction from {\sc Formula Game} played on formulas with the form $\exists v_1 \forall v_2 \exists v_3 \forall v_4 \dots \exists v_n \varphi$.
Players~1 and~2 from the GG game respectively mimic the $\exists$- and $\forall$-players. We construct the
formula's quantifier prefix by a chain of diamond gadgets (see Figure~\ref{fig: diamond}). Each diamond represents the choice that the
$\exists$-player, in case of an existential quantifier, or $\forall$-player, in case of a universal quantifier,
must make. The assignment of truth values for
variables 
is now simulated as described in Figure~\ref{fig: diamond}. 

\begin{figure}[htbp]
\includegraphics[width=\textwidth,scale=0.8]{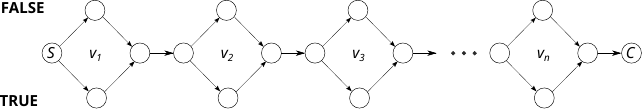}
\caption{Beginning at the designated start node~$S$, Player~1 has the first choice: marking the lower node sets $v_1$ to {\tt true}, while marking the upper node sets $v_1$ to {\tt false}. Both of these nodes have one outgoing edge, going to the right-hand node of the first diamond. Player~2 marks this node, after which Player~1 marks the first node of the next diamond. Player~2 can now choose to set $v_2$ to {\tt true} by marking the lower node, or to {\tt false} by marking the upper node. Play continues this way until Player~1 finally chooses for $v_n$ (since the last quantifier is~$\exists$).} \label{fig: diamond}
\end{figure}

After play has gone through the last diamond, we arrive at node~$C$, with Player~2 to move. In the Formula Game being simulated, all variables have been assigned a value by now and the players take the two final turns. To simulate these, we need two help gadgets: a \emph{clause chooser} gadget and \emph{literal chooser}
gadget. We 
construct the \emph{clause chooser} gadget by representing each clause with a node. Node~$C$ is given an outgoing edge to each of these nodes. Player 2 can now choose to continue play in any one of the clauses.

We construct the \emph{literal chooser} gadget by first adding a new node for each literal. We then add edges from each clause node to the nodes representing its literals. We also add an edge from each one of these literal nodes to the node corresponding to it within the diamond structures (i.e., the edges from nodes for literals~$v_i$ and~$\neg v_i$ respectively go to the lower and upper nodes in the diamond of~$v_i$).
In response to Player~2's choice of clause made in the \emph{clause chooser} gadget, Player~1 chooses a literal~$l$~from the chosen clause and marks its node. 

This node has only one outgoing edge, leading to a diamond-structure node that was previously marked in the game if and only if one of the players selected it during the earlier variable assignment simulation. If it is indeed already marked (the variable's assigned value makes the chosen literal true), Player~2 now has to mark it again and loses. If it is not already marked, Player~2 can mark it and force Player~1 to lose, as the only available node from here is the already marked node on the right side of the diamond.

So, with this construction, Player~1 has a winning strategy in the GG game if and only if the $\exists$-player has a winning strategy in the Formula Game. 
\end{proof}

\begin{theorem}[S.~Reisch~\cite{hex}] \label{thm: gg final version}
{\sc Generalized Geography} on graphs with maximum degree~3, maximum indegree~2 and maximum outdegree~2 is PSPACE-complete.
\end{theorem}

\begin{proof}
We can replace all nodes with indegree~3 by an equivalent chain of nodes with indegree~2 or less (see Figure~\ref{fig: indegree}). The treatment of nodes with outdegree~3 can be seen by reversing all arrows in the diagram. Indegree or outdegree above~3 can be handled similarly with a longer chain. 
The final case of indegree and outdegree exactly~2 (making total degree~4) is similarly straightforward. 
\begin{figure}[htbp]
\centering
\includegraphics[scale=1.35]{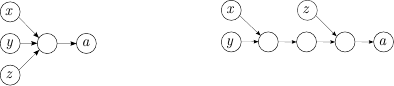}
\caption{A node with indegree~3 (left) is replaced by an equivalent chain of nodes with indegree~2 or less (right).} \label{fig: indegree}
\end{figure}
\end{proof}

Note that in PSPACE-hardness proofs of many other games, such as Go~\cite{go} and Hex~\cite{hex}, the reduction is done from a version of GG that is also assumed to be
planar. For our proof, however, this assumption is not necessary.\footnote{
\label{fn: planarity} It has been well known since~\cite[Fig. 8a]{hex} that GG is PSPACE-complete even when the given graph is assumed to be planar. For many games, this assumption is indispensable in a hardness proof,
as it is not possible to let two structures cross without interfering with each other. For that reason, proofs for such games require
a way to deal with any crossing edges in GG. Planarity ensures that these crossings simply never occur. Hive has a planar nature as well, except for the fact that the Grasshopper  (and technically, the Beetle)
can travel over structures. We use this to our advantage to overcome the issue of crossing edges. See Section \protect\ref{subsec: crossings} for details. }

We now define the problem of {\sc Formula Game Geography}, which is easily seen to be PSPACE-complete. In Section~\ref{sec: proof} we reduce this to {\sc Hive}, establishing our main result.

\begin{definition}\label{def: FGG}
{\sc Formula Game Geography (FGG)} is {\sc GG} played on a graph that is constructed by first applying
the reduction from Theorem~\ref{thm: GG PSH}
to a {\sc Formula Game} instance with the form specified in Theorem~\ref{thm: FG PSC}, then running the output of that through the reduction from Theorem~\ref{thm: gg final version}. 
That is, {\sc FGG} is {\sc GG} played on a graph that simulates a {\sc Formula Game} instance with the form $\exists v_1 \forall v_2 \exists v_3 \forall v_4 \dots \exists v_n \psi$ and satisfies the degree properties of Theorem~\ref{thm: gg final version}.
\end{definition}

\begin{corollary}\label{cor: FGG PSC}
{\sc FGG} is PSPACE-complete.
\end{corollary}

\section{Main proof} \label{sec: proof}

We now set out to reduce FGG to {\sc Hive}.
We begin with an overview of the proof structure and general idea. Then we introduce and explain the individual gadgets used in the proof. In Section~\ref{subsec: remarks one Hive rule}, we describe how the gadgets are connected together. Section~\ref{subsec: crossings} details how we handle crossings. Finally, in Section~\ref{subsec: completing} we make a few miscellaneous observations, including that the reduction can be carried out in polynomial time and that the constructed position is legally reachable under Hive rules. 

\subsection{Proof structure and general idea} \label{subsec: general idea}

Given is an FGG instance, which by definition corresponds to some quantified Boolean formula~$\varphi = \exists v_1 \forall v_2 \exists v_3 \forall v_4 \dots \exists v_n \psi$. That is, we are given a graph consisting of components that correspond to components of~$\varphi$ -- i.e., to quantifiers, variables, clauses, etc. 
It will therefore be convenient at times to identify the graph components (for example, a diamond structure of the type seen in the proof of Theorem~\ref{thm: GG PSH}) with the corresponding components of~$\varphi$ (in this example, a quantified variable).
In this way, our proof will somewhat resemble a proof by reduction directly from {\sc Formula Game}, though we also take advantage of the degree properties provided in Theorem~\ref{thm: gg final version}.

From the given graph that represents~$\varphi$, we will construct a Hive position from which the game will proceed through a series of four stages, corresponding to the stages of play seen in the reduction from {\sc Formula Game} to GG (Theorem~\ref{thm: GG PSH}).
Black will represent the $\exists$-player (Player~1) and White will represent the $\forall$-player (Player~2).

First, in the \emph{quantifier} stage, the players simulate taking turns assigning truth-values to the variables. Next, the $\forall$-player chooses a clause in the \emph{clause chooser} stage. Then, in the \emph{literal chooser} stage, the $\exists$-player chooses a literal $l$ from that clause. Finally, the \emph{tester} stage consists of a sequence of forced moves that simulate
evaluating the chosen literal's assigned value. Black will win if and only if this literal~$l$ was originally assigned {\tt true}.

A main feature in the design of our simulation is that  all pieces are locked down at all times, except for one. 
So the player whose turn it is can
only move one 
piece. Depending on the situation, the piece can move to just one space or to a choice between two spaces.

We regulate the game flow
mostly
by locking and unlocking pieces
through repeated use of the One Hive rule. That this works in the case of any single gadget can be easily verified. However, it is not trivial that pieces that would be locked by the One Hive rule in individual gadgets remain locked when the gadgets are connected 
together.
The details of
this matter are treated in Section~\ref{subsec: remarks one Hive rule}. 

\begin{theorem}[Main theorem] \label{thm: hive}
Hive is PSPACE-hard.
\end{theorem}

\begin{proof}
We reduce FGG to Hive. To simulate the components of an FGG graph, we construct gadgets for the following: choice of truth-value at a quantifier, $\forall$-player's choice of clause, $\exists$-player's choice of literal, truth-value testing of the chosen literal, and joining two edges. The first and fourth are handled by a single gadget, the \emph{quantifier/tester}. The second and third are handled respectively by the \emph{clause chooser} and \emph{literal chooser} gadgets. The last is handled by the \emph{join}. For connecting these gadgets together, we also have the following structures: \emph{turn switcher}, \
\emph{direction changer}, and \emph{gap}. All constructions can be rotated in~$60^\circ$ increments when needed.

Before we present the gadgets, recall that dead exterior pieces must always be placed such that the top Beetle has no opportunity to create short-term threats in one step (such as by freeing a dangerous Ant that can win immediately). 
By inspection, it is straightforward to confirm that this condition is clearly met within each gadget. After Sections~\ref{subsec: remarks one Hive rule}-\ref{subsec: crossings}, it is also straightforward to verify
that
this condition holds
when
gadgets are connected together.

\subsection{Quantifier/tester} \label{subsec: quantifier}
This gadget is the heart of the reduction. It simulates both the choice a player has to make when assigning a value to a variable and the test that checks whether the variable's chosen value makes the later chosen literal true, which decides who wins. Figure~\ref{fig: E-quantifier spider untrapper} depicts the gadget for existentially quantified variables~$v_1, v_3, v_5,\dots, v_n$. For universally quantified variables~$v_2,v_4,\dots,v_{n-1}$, the construction is the same with colors reversed. The gadget's main piece is the Spider~$S$, which begins locked by the Freedom to Move rule but soon becomes freed to move either down or up. The choice of where it moves determines whether the simulated variable is assigned {\tt true} or {\tt false}. 

\begin{figure}[htbp]
\centering
\includegraphics[angle=-90,width=.8143\textwidth,draft=false]{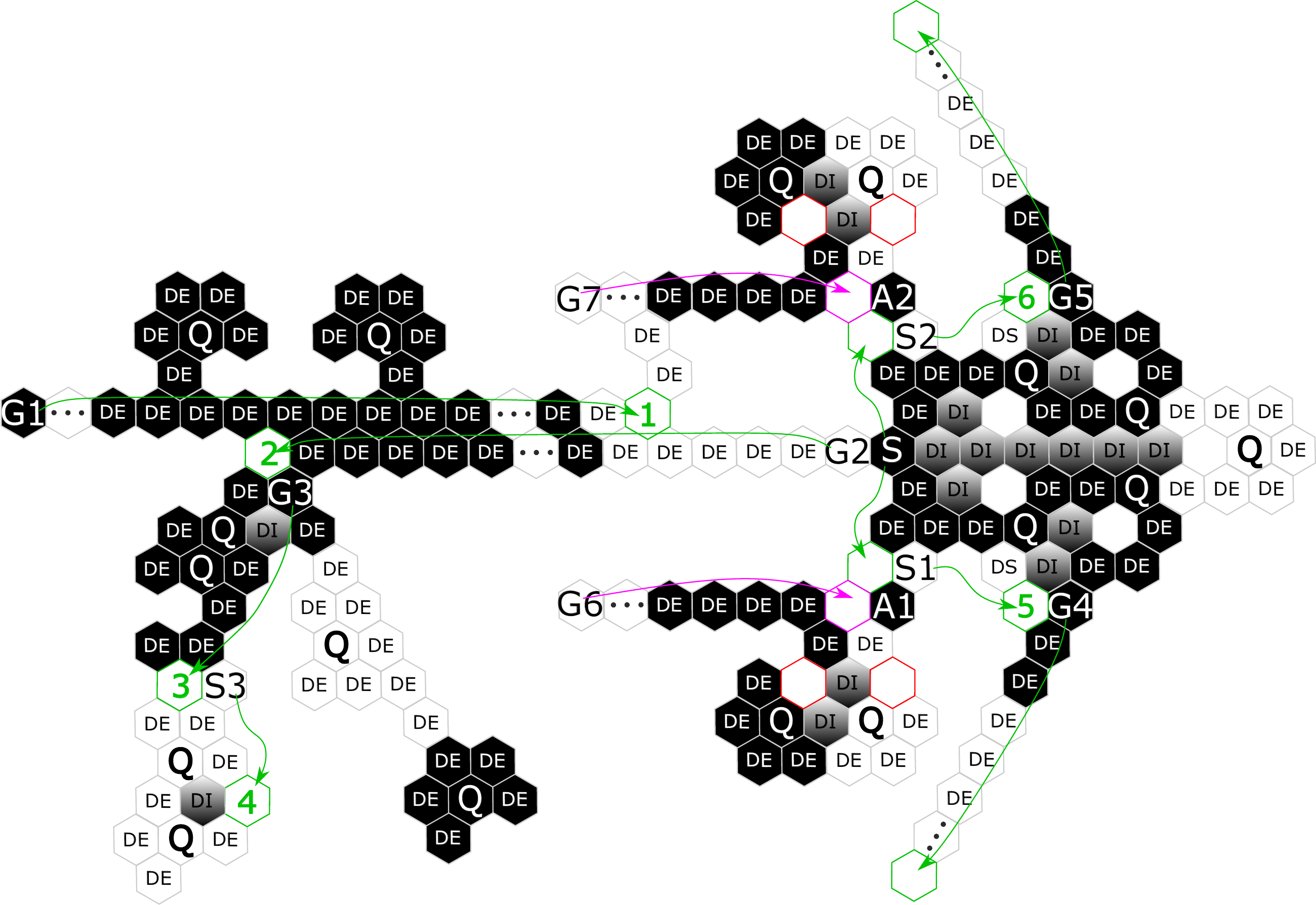}
\caption{A \emph{quantifier/tester} gadget for $\exists$-quantifiers. This gadget plays two roles in the FGG simulation. First, it simulates the choice given to the $\exists$-player for the quantification of one variable~$v_i \in \{v_1, v_3, v_5, \dots v_n\}$. Later on, after the \emph{literal chooser} stage has passed and the $\exists$-player has chosen a literal from the clause selected by the $\forall$-player, if that literal happens to be either~$v_i$ or~$\neg v_i$
then this gadget simulates the \emph{tester} that checks whether the assignment satisfies the chosen literal. Otherwise, the test is performed by a different 
\emph{tester} gadget, corresponding to the appropriate variable. The $\exists$-player then wins if and only if the chosen literal is indeed satisfied (see Lemma~\protect\ref{lem: first quantifier}).}
\label{fig: E-quantifier spider untrapper}
\end{figure}

All pieces in the diagram begin locked, except Grasshopper~1, which is located externally and is not part of this gadget. In the case of variables~$v_2, \dots, v_n$, Grasshopper~1 enters from the previous \emph{quantifier/tester gadget}. In the case of~$v_1$, Grasshopper~1 stands alone. When the simulation starts, this latter Grasshopper is the only free piece. In general, we use hexes with ``$\dots$'' to denote a lengthy chain of dead exterior pieces, with whatever is on the far side being understood as separate from the structure in the figure.

We now describe the flow of play in the \emph{quantifier} stage.
Without loss of generality, 
consider an existentially quantified variable~$v_i$. Black Grasshopper~1 comes in from outside to land on space~\textcolor{nicegreen}{1}. This frees the white Grasshopper~2, which can only jump away to space~\textcolor{nicegreen}{2} (the alternative is selfmate). The central black Spider $S$ is now freed, as desired, but the white Grasshopper~2 is currently threatening to surround a black Queen. Since White's last move freed black Grasshopper~3, Black moves it to space~\textcolor{nicegreen}{3}, locking White's Grasshopper and preventing the win. This then frees white Spider~3, which moves to space~\textcolor{nicegreen}{4} (elsewhere is selfmate), simultaneously locking Grasshopper~3 under the One Hive rule.  Now it is Black's turn, with the black Spider~$S$ free to move. 
Note that white Spider~3 will never return to its previous space in the future, because black Grasshopper~3 would then become able to surround a white Queen.

Every piece is now locked down except $S$,
whose movement will now determine the truth-value of~$v_i$.
This includes white Grasshoppers~6 and~7, which are part of another structure. Later, after the \emph{quantifier,} \emph{clause chooser,} and \emph{literal chooser} stages have been completed, it is possible that white Grasshopper~6 or~7 will become unlocked and enter the present gadget in the \emph{tester} stage. As we will see in Claim~\ref{claim: tester grasshoppers}, this will occur if and only if the literal~$l$ chosen during the \emph{literal chooser} stage is either~$v_i$ or~$\neg v_i$, respectively.

Since the black Spider~$S$ is the only 
currently 
movable piece, the $\exists$-player chooses either to move~$S$ down (assigning~$v_i$ {\tt true}) or up (assigning it {\tt false}). Without loss of generality, suppose it moves down. This frees white Spider 1, which has no choice but to move to space~\textcolor{nicegreen}{5} (the alternative is selfmate). This then frees black Grasshopper~4.

We claim now that Grasshopper~4 can only jump away to the next structure. If it jumps left instead, Black loses. To see this, recall that the special dead piece~\textbf{DS} consists of a white Beetle on top of
three black Beetles, with a black Queen underneath. This tower of pieces has two roles. In the first place, naturally, it serves as a dead piece. This is because the top white Beetle can make no valuable move in one turn, and if it ever would move, the black Beetle under it would be able to climb on top of it, leaving the next black Beetle free to find a white Queen in the next few moves to surround without any recourse. In the second place, the tower ensures that black Grasshopper~4 cannot jump over white Spider~1, lest the top white Beetle suddenly step left and surround the black Queen inside the tower. So the Grasshopper must move southwest to the next structure, unlocking whatever is there.

Having described the flow of play in the \emph{quantifier} stage, we skip the \emph{clause chooser} and \emph{literal chooser} stages for the time being and proceed directly to the \emph{tester} stage. Later we discuss the intermediate stages.

\begin{lemma}\label{lem: first quantifier}
If a literal~$l \in 
\{v_i,\neg v_i\}$ is the one chosen during the \emph{literal chooser} stage, the $\exists$-player (Black) has a winning strategy in the \emph{tester} stage if and only if~$l$ is true under the variable assignment that was chosen during the \emph{quantifier} stage.
\end{lemma}

\begin{proof}
The proof of this lemma relies on the following claim. 
\begin{claim}\label{claim: tester grasshoppers}
In the \emph{tester} stage, a Grasshopper enters into the lower (respectively, upper) \emph{tester} entrance for a given $\exists$-quantified variable~$v_i$ if and only if~$l$ is~$v_i$ (respectively,~$\neg v_i$).
On the other hand, a~Grasshopper enters into the lower (respectively, upper) \emph{tester} entrance for a given $\forall$-quantified variable~$v_i$ if and only if~$l$ is~$\neg v_i$ (respectively,~$v_i$).
\end{claim}
We verify this claim later, in Section~\ref{subsec: completing}. For now, assume the claim is true and consider first the case in which variable~$v_i$ was quantified existentially. 
Suppose the chosen literal~$l$ is indeed either~$v_i$ or~$\neg v_i$, so white Grasshopper~6 or~7, respectively, has therefore entered the \emph{tester.} We now show that the $\exists$-player (Black) has a winning strategy if and only if~$l$
is true under the value assigned earlier to~$v_i$ in the \emph{quantifier} stage. There are four subcases to consider.
\begin{enumerate}
\item \label{case: QT1}
Suppose~$v_i$ was assigned {\tt true} at the \emph{quantifier} stage and the $\exists$-player chose~$l=v_i$ at the \emph{literal chooser} stage. So Spider~1 is on space~\textcolor{nicegreen}{5} and the black Grasshopper~4 is out of the picture.
As the \emph{tester} stage begins, white Grasshopper~6 enters the gadget on the indicated space (by Claim~\ref{claim: tester grasshoppers}), freeing
black Ant~1 to surround
the nearby 
white Queen by moving two hexes down. Thus the $\exists$-player has a winning strategy.

\item \label{case: QT2} Suppose $v_i$ was assigned {\tt true} at the \emph{quantifier} stage and the $\exists$-player chose $l = \neg v_1$ at the \emph{literal chooser} stage. In this case, Spider~1 and Grasshopper~4 are placed as in Case~\ref{case: QT1}, but white Grasshopper~7 enters the gadget rather than~6. When it does, no piece is freed. So Black has no way to stop White from winning on the next turn by jumping Grasshopper~7 northwest, or just punishing Black for moving one of the dead pieces. Thus the $\forall$-player has a winning strategy, as desired.
\item The case in which~$v_i$ was assigned {\tt false} and~$\neg v_i$ was chosen is symmetric to Case~\ref{case: QT1}.    
\item The case in which~$v_i$ was assigned {\tt false} and~$v_i$ was chosen is symmetric to Case~\ref{case: QT2}.
\end{enumerate}
We now consider universally quantified variables. Recall that here the colors of pieces are reversed in the gadget. So the Queens being surrounded in the four subcases are of the opposite color, reversing the outcomes. However, Claim~\ref{claim: tester grasshoppers} provides that the Grasshopper arriving in the \emph{tester} stage enters through the opposite entrance from before, reversing the outcomes again. So the four subcases again produce the correct winner.
\end{proof}

\subsection{Join} \label{subsec: join}
The \emph{join} gadget simulates an FGG node with indegree~2 and outdegree~1, joining two other structures together. We require these after each \emph{quantifier} gadget in order to combine the two outgoing Grasshopper paths into one, which then feeds into the next \emph{quantifier} (or, in the case of the final \emph{quantifier,} the first \emph{clause chooser} gadget). We also need \emph{join} gadgets when a literal appears in more than one clause, so that its multiple occurrences can be combined and fed into the appropriate \emph{tester}. This is because only one Grasshopper can enter a \emph{tester} gadget in the \emph{tester} stage, so if there are multiple Grasshopper representatives of the same literal coming from different clauses in~$\psi$, their paths must
be joined
beforehand.

Figure~\ref{fig: join} depicts the \emph{join} gadget. To start, either Grasshopper 1, coming from another structure northwest, or Grasshopper 2, coming in from southwest, lands on the corresponding space marked~\textcolor{nicegreen}{1}. This frees the black Grasshopper, which can only jump safely to space~\textcolor{nicegreen}{2} in the east. This then frees Grasshopper~3, which continues further east to the next structure. 

\begin{figure}[htbp]
\centering
\includegraphics[scale=0.335,draft=false]{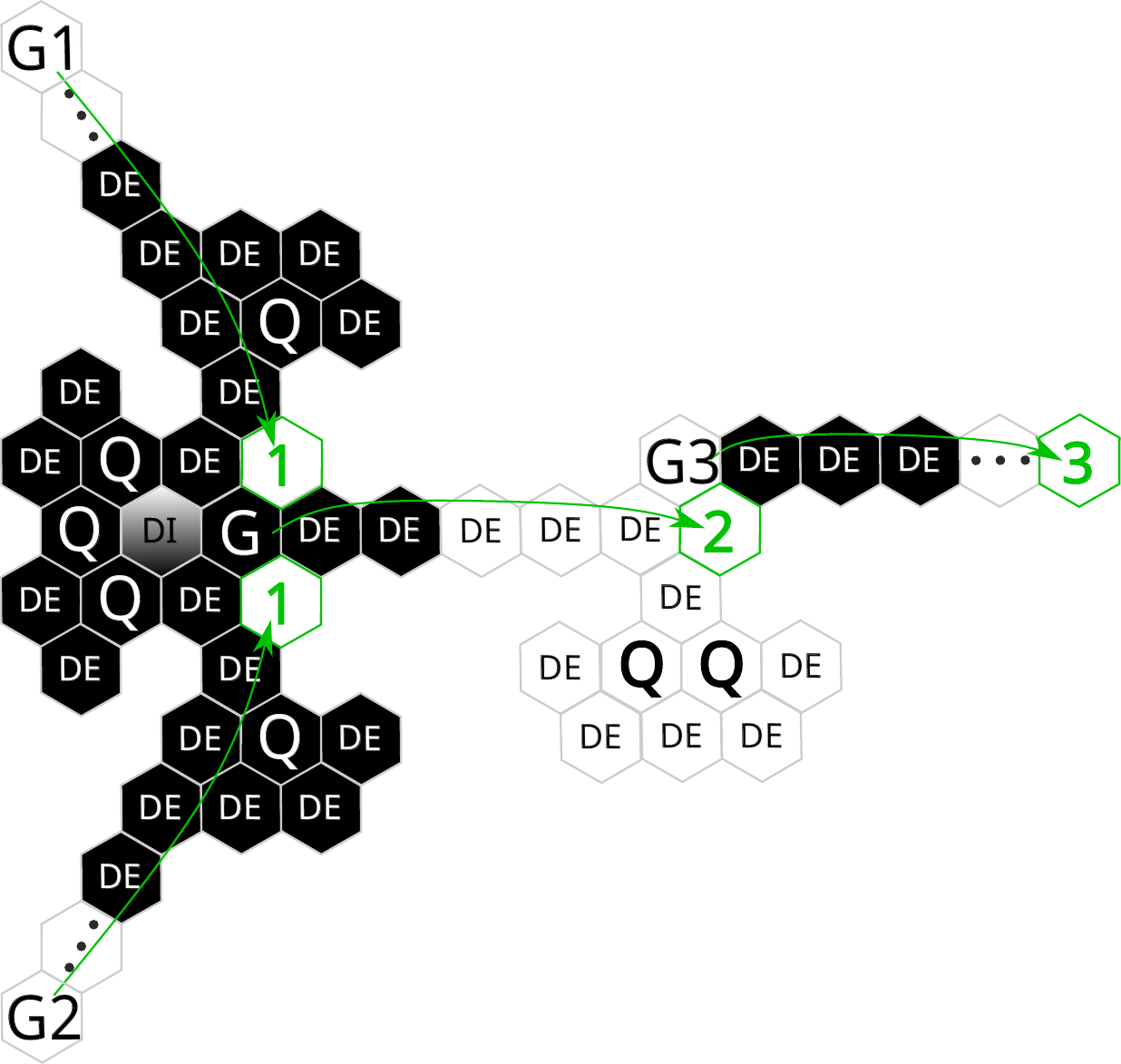}
\caption{A white-to-black \emph{join} gadget. (Black-to-white \emph{joins} are identical with colors reversed.)}
\label{fig: join}
\end{figure}

\enlargethispage{-\baselineskip}
\subsection{Turn switcher} \label{subsec: turn switcher}
This gadget's only function is to switch turns.
We have it available for whenever 
a black Grasshopper
is exiting one structure while a white Grasshopper is needed to enter the next, or vice versa. 
This occurs, e.g., between two \emph{clause choosers} (see Section~\ref{subsec: clause/literal chooser}). 
Figure~\ref{fig: switch} depicts the black-to-white version.
Here, the incoming black Grasshopper lands on space~\textcolor{nicegreen}{1}. This frees the white Grasshopper, which has no safe choice but to continue the same direction. 

\begin{figure}[htbp]
\centering
\includegraphics[scale=0.335,draft=false]{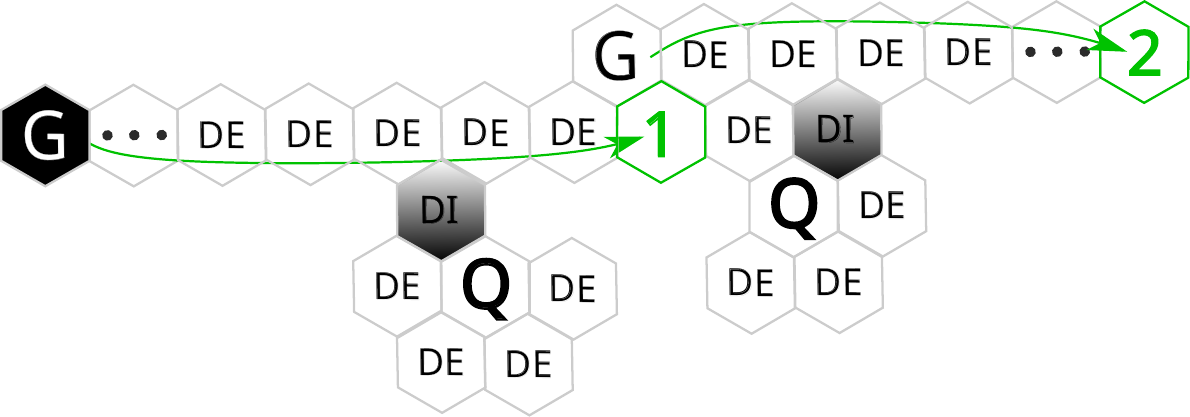}
\caption{A white-to-black \emph{turn switcher}. (The black-to-white are identical
with colors reversed.)}
\label{fig: switch}
\end{figure}

\subsection{Clause/literal chooser}\label{subsec: clause/literal chooser}
The \emph{clause chooser} gadget, shown in Figure~\ref{fig: chooser-changer}, functions as an FGG node with indegree~1 and outdegree~2. The black Grasshopper enters from the left and lands on space~\textcolor{nicegreen}{1}. This frees the white Grasshopper, which can jump to one of the two spaces marked~\textcolor{nicegreen}{2}. This represents White's choice between two clauses. 
The \emph{literal chooser} gadget, constructed the same with colors reversed, 
represents the similar choice of the $\exists$-player (Black) between two literals.

Chaining several \emph{clause choosers} together in the manner described in Theorem~\ref{thm: gg final version}, with \emph{turn switchers} in between to change back colors, effectuates giving White a choice over many clauses. Likewise for Black and literals in the next stage. If formula~$\psi$ has  $q \in \mathbb N$ clauses and $r \in \mathbb N $ literals, the reduction requires $q-1$ \emph{clause chooser} and $r-1$ \emph{literal chooser} gadgets.

\begin{figure}[htbp]
\centering
\includegraphics[scale=0.335,draft=false]{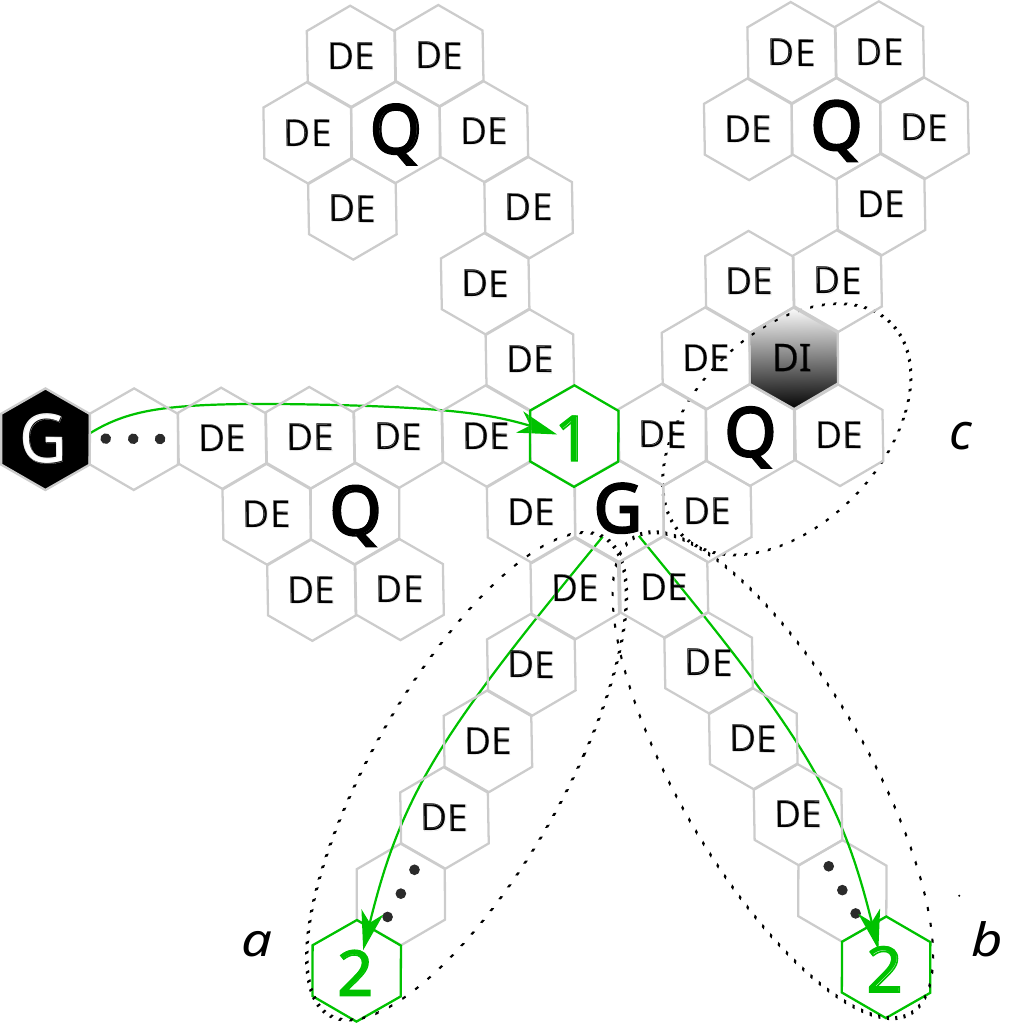}
\caption{Combined diagram for \emph{clause chooser, literal chooser, and direction changer} gadgets. The four pieces in group~$c$ are not needed in a \emph{clause chooser.} Reversing all piece colors of a \emph{clause chooser} yields a \emph{literal chooser.} Removing group~$a$ from the diagram yields a black-to-white~$60^\circ$ clockwise \emph{direction changer}. For this, group~$c$ is necessary because its absence would make the white Grasshopper's jump illegal under the One Hive rule. Finally, a black-to-white~$120^\circ$ clockwise \emph{direction changer} is constructed as a \emph{clause chooser} without group~$b$ (here, again,~$c$ is not needed).} 
\label{fig: chooser-changer}
\end{figure}

\enlargethispage{-\baselineskip}
\subsection{Direction changers} 
These gadgets are functionally \emph{turn switchers} that change the direction of the game's flow of play by~60 or~120 degrees. Figure~\ref{fig: chooser-changer} shows either a~$60^\circ$ or~$120^\circ$ black-to-white clockwise \emph{direction changer,} depending on which piece groups are included. Piece colors can naturally 
be reversed to produce white-to-black \emph{direction changers.} Play proceeds as in the \emph{clause} and \emph{literal choosers,} except with no choice to make for the outgoing Grasshopper's direction.

\emph{Direction changers} are needed in many places, such as between the \emph{quantifier} gadgets, whose outgoing Grasshoppers always move northwest or southwest and must be redirected so they can enter a \emph{join} (see Figure~\ref{fig: Quantification-Sketch}).

\begin{figure}[htbp]
\centering
\includegraphics[width=\textwidth,draft=false]{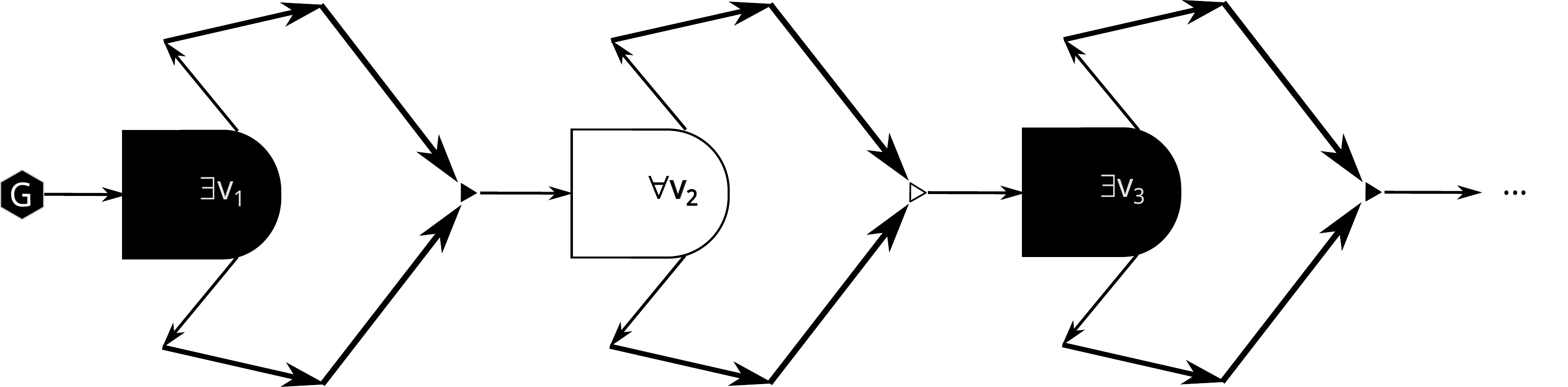}
\caption{Sketch of \emph{quantifier} gadgets connected together.~$G$ represents the starting Grasshopper. The arrows directed northwest and southwest out of the \emph{quantifiers} represent the chains of dead pieces over which the outgoing Grasshoppers may jump. The points where one arrow touches another contain~$60^\circ$ or~$120^\circ$ \emph{direction changers}. The small triangles between \emph{quantifiers} denote \emph{join} gadgets.} 
\label{fig: Quantification-Sketch}
\end{figure}

\subsection{Gap} \label{subsec: gap}
\enlargethispage{\baselineskip}
To ensure that all pieces locked by the One Hive rule are truly locked, it is necessary at some places to disconnect two structures by including 
a \emph{gap.} There are three types of places where we need a \emph{gap.} First, a \emph{gap} is required in one of the two paths between two \emph{quantifier} gadgets, to ensure that the Hive is connected in only one place. Second, we need a \emph{gap} wherever a \emph{literal chooser} (or a \emph{join} that has connected two \emph{literal choosers}) would meet a \emph{tester} gadget. Third, a \emph{gap} is employed whenever a crossing occurs. The first two uses of \emph{gaps} are discussed in Section~\ref{subsec: remarks one Hive rule}. The third is treated in Section~\ref{subsec: crossings}.

In Figure~\ref{fig: gap}, white Grasshopper 1 comes in from the previous structure to land on space~\textcolor{nicegreen}{1}. This frees the black Spider, which can pass over the gap by following the arrow and move to space~\textcolor{nicegreen}{2} (the alternative is selfmate). This move frees white Grasshopper 2, which jumps away to the next structure (moving elsewhere loses; note that jumping over the Spider would land on a space surrounding a~\textbf{DE} beetle tower that contains a white Queen).

\begin{figure}[htbp]
\centering
\includegraphics[scale=0.35,draft=false]{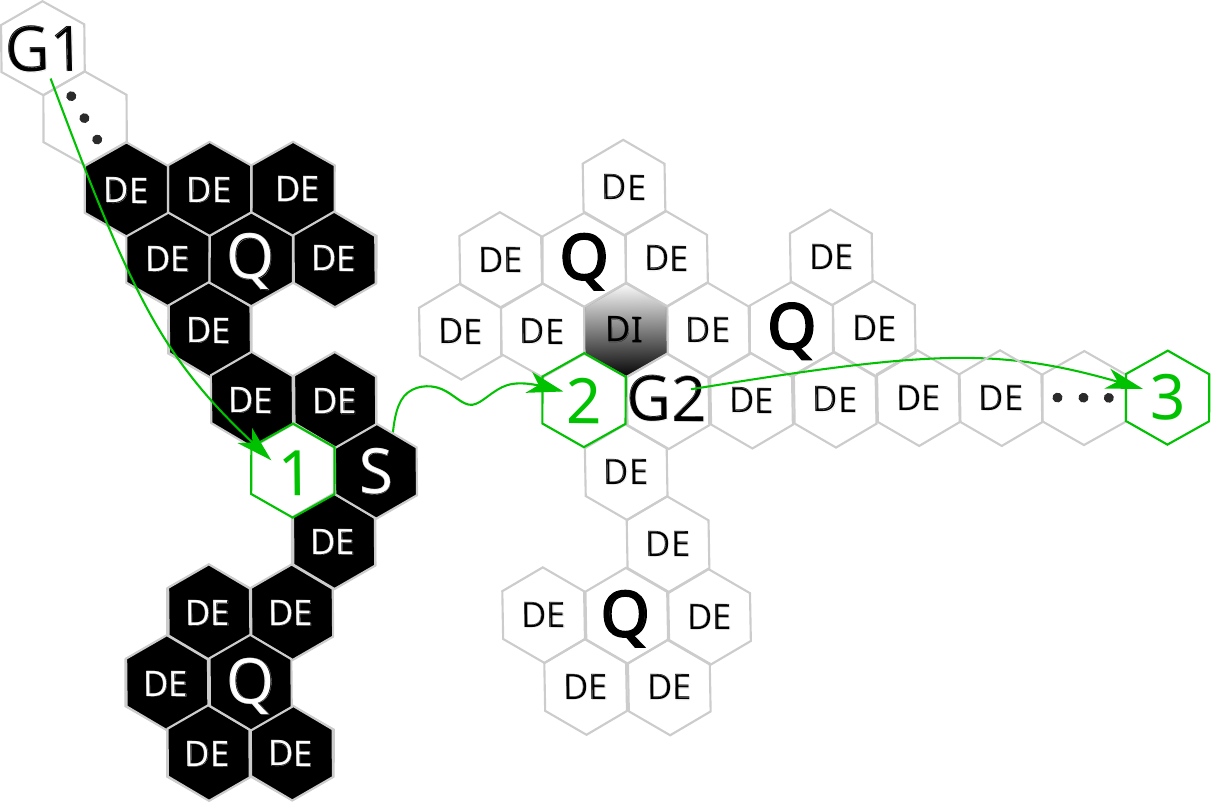}
\caption{A white \emph{gap} gadget to disconnect parts of
the Hive. Colors may be reversed.}
\label{fig: gap}
\end{figure}

\subsection{Connecting the gadgets} \label{subsec: remarks one Hive rule}
In this section, we show how the gadgets are connected to simulate an FGG graph, while ensuring that every gadget is linked to every other gadget by exactly one path of adjacent pieces. There must be at least one path because of the One Hive rule, and there must be at most one path to make sure pieces purposely locked by the One Hive rule are indeed locked.

When the gadgets are connected, the full board configuration of the FGG simulation can be seen as beginning with a long horizontal line of \emph{quantifier/tester} gadgets -- see Figure~\ref{fig: Quantification-Sketch}. Play flows from west to east. A structure connects to another through a single chain of dead pieces over which an outgoing Grasshopper may jump. Each of the \emph{quantifier} gadgets is linked to \emph{direction changers} that eventually feed into a \emph{join}. A \emph{gap} is also included among the \emph{direction changers} in the lower of the two paths coming from each \emph{quantifier}.
The \emph{join} connects to the next \emph{quantifier} gadget, except for the \emph{join} coming from the rightmost one. This last \emph{join} sends its output Grasshopper into a tree of further structures that extend above and below the horizontal line. The tree consists of the \emph{clauser chooser} gadgets, arranged as described in Section~\ref{subsec: clause/literal chooser} (based on the construction from Theorem~\ref{thm: gg final version}), which then further branch into \emph{literal chooser} gadgets that are themselves
arranged the same way. That is, the \emph{clause choosers} together form a chain of binary choices that collectively simulate a single choice of one clause among many, each of which in turn branches into a chain of literal choosers that do the same for literals in that clause.
So, for each combination of clause~$c_i$ and literal~$l_j \in c_i$, there is one path. At the end of each of these paths, a \emph{gap} is inserted.

\enlargethispage{-.7\baselineskip}
The leaves coming out of the \emph{literal choosers} feed into paths heading back west via \emph{direction changers}, each leading toward the appropriate \emph{quantifier/tester} for the literal in question.
Paths from \emph{literal choosers} corresponding to the same literal in different clauses are connected through \emph{join} gadgets before reaching the destination \emph{tester}.
The set of \emph{join} gadgets that feed into any given \emph{tester} entrance form a sort of branching tree themselves, if viewed in the direction opposite to the flow of play.

When connecting the path from a literal~$l$ to the corresponding \emph{tester}, if~$l \in \{v_i, \neg v_i\}$ for~$i \in \{1, 3, 5, \dots, n\}$ (recall that this means the variable~$v_i$ is existentially quantified), we connect it via the lower (true) Grasshopper entrance if and only if $l = v_i$. If~$l \in \{v_i, \neg v_i\}$ for~$i \in \{2, 4, 6, \dots, n-1\}$ (that is, the variable is universally quantified), we do the reverse, connecting instead via the lower (true) entrance if and only if $l = \neg v_i$. This is to ensure that Claim~\ref{claim: tester grasshoppers} holds.

With the construction described in this section, we see that gadgets are connected in such a way as to satisfy the One Hive rule at all times, while ensuring that pieces intended to be locked down by the rule remain duly locked. The only consideration not yet covered is the possible intersection of two paths, which we treat in Section~\ref{subsec: crossings}. 

\subsection{Crossings} \label{subsec: crossings}
Recall from fn.~\ref{fn: planarity} that we do not assume planarity of the input GG graph for our reduction. In this section, we present a way to deal with crossings in Hive directly.

When enough space is taken between the gadgets, no crossings occur in the \emph{clause chooser} and \emph{literal chooser} areas. In the area after the incorporated \emph{gaps}, however, when the literals from the clauses are directed to the corresponding \emph{testers}, crossings may occur for some formulas. Recall that some of the paths out of \emph{literal choosers} combine through \emph{joins} before they are directed to the corresponding \emph{tester}. In this process, too, crossings might occur. 

To ensure that these crossings cause no problems, first we describe how the crossings themselves are implemented. Then, we show why they 
do not 
free any pieces we need to keep locked under the One Hive rule.

The crossing itself is easily constructed. At the point of intersection, each of the two paths is a chain of dead pieces over which a Grasshopper can jump from one gadget into the next. So, to
implement the crossing we can simply allow the two chains to cross each other, with a single piece at the intersection serving as a common element to both chains. The Grasshoppers of both chains are still able to jump over their respective chains in the same manner as without a crossing.

To ensure pieces in the Hive locked by the One Hive rule remain locked, we simply add a \emph{gap} for each crossing, placed in one of the crossing paths
between the point of intersection 
and
the
\emph{tester} gadget. We now argue that this approach suffices.

\begin{lemma}\label{lem: crossing lemma}
Even with crossings, all gadgets still connect via exactly one path.
\end{lemma}

\begin{proof}
Let~$A$ be a path to a \emph{tester}~$t_a$ from one of the incorporated \emph{gaps.} Call this \emph{gap}~$g_a$. Let~$B$ be a path to a \emph{tester}~$t_b$ from another \emph{gap}~$g_b$. Note that all \emph{testers} are connected along the horizontal line of \emph{quantifier/tester} gadgets, so they are a single structure~$S$ for the purposes of the One Hive rule.

Suppose~$A$ and~$B$ cross. Call the shared hex~$c$. By our construction, one of the paths (without loss, path~$A$) is constructed as normal between the crossing point~$c$ and the \emph{tester}~$t_a$, while the other,~$B$, has an additional \emph{gap} between~$c$ and~$t_b$. Were it not for this added \emph{gap,} the reduction would break, as many pieces throughout the Hive would then be unhampered by the One Hive rule -- namely, those between~$c$ and~$t_a$, those
between~$c$ and~$t_b$, and a large number of pieces within 
gadgets between~$t_a$ and~$t_b$ (inclusive) in the horizontal line. This is because anything connected to the Hive through these pieces would have a secondary path of connection on top of the regular one, thanks to the connection between~$t_a$~and~$t_b$ through~$c$.

However, because of the added \emph{gap} in~$B$ between~$c$ and~$t_b$, we find the following. First, each of the pieces in the fragment of~$B$ running between this gap and~$t_b$ has only one path connecting it to the Hive -- namely, via the fragment itself. Second, each piece in~$B$ between~$c$ and the gap has only one path as well -- namely, via the fragment of~$A$ from~$c$ to~$t_a$. The latter fragment is also the only connecting path to~$S$ for pieces between~$g_a$ and~$c$ and for pieces between~$g_b$ and~$c$. All other gadgets are connected to the Hive as though~$A$ and~$B$ never crossed.

Since all gadgets connect by exactly one path, the lemma is proved.
\end{proof}

\subsection{Completing the proof} \label{subsec: completing}
Having laid out the reduction, we now make a few final observations.
First, we are now finally equipped to verify Claim~\ref{claim: tester grasshoppers}.
By inspection of the constructed position as a whole, we see that a Grasshopper indeed eventually enters an existential \emph{quantifier/tester} gadget on its northern (respectively, southern) side if and only if the corresponding literal~$\neg v_i$ (respectively,~$v_i$) is the literal~$l$ chosen in the \emph{literal chooser} stage. We similarly see the reverse for the universal case.

Next, we observe that the starting position of our FGG simulation is legally reachable from the starting state of the game. Recall that Hive always begins with an empty playing field. Our players will take turns placing pieces and moving them, starting with the first \emph{quantifier} and working their way through all the gadgets. Many pieces can be placed exactly where they need to end up in the position and never moved. Others can be placed nearby and then moved as needed. Given the pieces' versatility, it is easy to see how the starting position of the simulation can be reached. 
The Beetles are useful pieces to assist another piece in reaching its spot (for instance, by temporarily making an extra connection to free a piece locked down by the One Hive rule), as they can reach every possible space and are always able to return to their designated location. Any excess pieces in the players' supply are placed on the board and moved into a chain, attached to a gadget lying on the border of the Hive in such a way that it does not influence the game, with a dead exterior piece at the end of the chain.

Finally, we note that the size of the construction is polynomial relative to the size of the FGG graph. This is straightforward, as each gadget consists of a constant number of pieces, and each quantifier, clause, and literal requires a limited number of gadgets to be simulated.

This concludes the proof of the main theorem.\end{proof}

\section{Conclusion} \label{sec: conclusion}
We've shown that the winning strategy decision problem for Hive generalized to arbitrary numbers of pieces is PSPACE-hard by reduction from {\sc Formula Game}, with the aid of {\sc Formula Game Geography}, a variant of {\sc Generalized Geography} played 
on
graphs that simulate {\sc Formula Game} instances and have maximum indegree~2, outdegree~2, and total degree~3. Our proof has followed the broad reduction strategy of~\cite{go}, but with peculiarities owing to the hexagonal geometry of Hive and unique challenges imposed by the One Hive rule. At the same time, the One Hive rule has also been a great advantage in constructing the gadgets, by enabling us to prevent the majority of pieces from moving and control the flow of the game.

The obvious recommendation for future research is to determine the exact upper and lower bounds of {\sc Hive}'s complexity. As mentioned in Section~\ref{sec:introduction}, Hive has no correlate of the 50-move rule in chess; in fact, outside of the simultaneous surrounding of Queens,  Hive has no official way for games to end in a draw other than by
player agreement. Accordingly, we have conjectured that {\sc Hive} is EXPTIME-complete, a proof of which we hope will be found in future work.

Another recommendation, as mentioned in Section~\ref{sec:introduction}, is to consider a generalization of Hive in which the players have only one Queen each. This would align the generalized game more closely to the spirit of the original game, much as~\cite{FraenkelLichtenstein} and~\cite{storer} maintain one King per player in generalized chess. Our reduction makes essential use of the presence of multiple Queens by employing them to create selfmate hexes, which we use to prevent pieces
from traveling in certain directions. Without this resource, an alternative may be to create ``indirect'' selfmate hexes by placing Ants in such a way that they are locked down by the One Hive rule until an enemy piece moves to the critical hex (what would have been a selfmate hex in our current reduction). That Ant would then become free to travel any distance and, if the construction is successful, reach the enemy Queen to surround it. Another approach is to make use of repeated game-winning threats to force play, as in~\cite{go}, rather than relying on the One Hive rule and selfmate hexes.

Along similar lines, as noted in fn.~\ref{fn: all Queens}, even the assumption of multiple Queens allows for more than one possible generalization of the standard game. We expect that a variant of $n$-Hive in which a player must surround all enemy Queens to win, rather than just one, is at least as hard as other variants. However, the gadget constructions employed in the present paper are not suited to show this, given their reliance on selfmate hexes.

Finally, we mention that as an addition to the basic game, the publishers of Hive have released three other pieces in expansions between 2007 and 2013: the Ladybug, Mosquito, and Pillbug. 
In principle it could be investigated whether the inclusion of these 
additional pieces brings any increase to the game's computational complexity. 
We conjecture that it does not, as the complexity of Hive appears to lie to a greater extent in the geometry of the game and its abstract rules, such as the One Hive rule, rather than the exact movement of the pieces. Nevertheless, it cannot currently be ruled out.

\enlargethispage{-\baselineskip}
More interestingly, perhaps, we wonder whether a hardness proof can be carried out for Hive with a \emph{smaller} piece set than the standard one. While some pieces seem indispensable -- for example, Beetles -- the Ants are seldom used and appear replaceable in our reduction, and we further ask whether a version of Hive without access to Spiders might also prove still to be computationally hard.

\bibliographystyle{plain}


\end{document}